\documentclass[letterpaper, 10 pt, conference]{ieeeconf}

\IEEEoverridecommandlockouts    

\usepackage{mathptmx} 
\usepackage{times} 

\usepackage{amsmath} 
\usepackage{amssymb}  
\usepackage{amsthm}
\usepackage{amsxtra}
\usepackage{amsfonts}
\usepackage{mathtools}

\usepackage{mathrsfs}

\usepackage{graphicx}

\usepackage{color}

\usepackage{pict2e}

\usepackage{hyperref}
\hypersetup{
	colorlinks=true,
	linkcolor=black,
    citecolor=black,
	filecolor=magenta,
	urlcolor=cyan,
}

\usepackage[backend=biber,style=ieee]{biblatex}
\usepackage{bbm}
\DeclareSymbolFont{bbold}{U}{bbold}{m}{n}
\DeclareSymbolFontAlphabet{\mathbbm}{bbold}

\DeclareMathAlphabet{\mathcal}{OMS}{cmsy}{m}{n}

\makeatletter
\DeclareFontEncoding{LS1}{}{}
\DeclareFontSubstitution{LS1}{stix}{m}{n}
\DeclareMathAlphabet{\mathscr}{LS1}{stixscr}{m}{n}
\makeatother

\newcommand{\trs}{\prime}
\DeclareMathOperator*{\argmin}{arg\,min}

\newtheorem{theorem}{Theorem}
\newtheorem{assumption}{Assumption}

\newtheorem{proposition}{Proposition}
\newtheorem{lemma}{Lemma}
\newtheorem{definition}{Definition}
\newtheorem{example}{Example}
\newtheorem{remark}{Remark}

\title{\LARGE \bf On Riccati contraction in time-varying linear-quadratic control}

\author{Jintao Sun and Michael Cantoni
\thanks{This work was supported by a Melbourne Research Scholarship and the Australian Research Council (DP210103272).}
\thanks{J. Sun is a PhD student at the Department of Electrical and Electronic Engineering, The University of Melbourne, Melbourne, Australia
        {\tt\small jintaos@student.unimelb.edu.au}}%
\thanks{M. Cantoni is with the Department of Electrical and Electronic Engineering, The University of Melbourne, Melbourne, Australia
        {\tt\small cantoni@unimelb.edu.au}}%
}

\begin{document}

\maketitle
\thispagestyle{empty}
\pagestyle{empty}

\begin{abstract}
Contraction properties of the Riccati operator are studied within the context of non-stationary linear-quadratic optimal control. A lifting approach is used to obtain a bound on the rate of strict contraction,  with respect to the Riemannian metric, across a sufficient number of iterations. This number of iterations is related to an assumed uniform controllability and observability property of the dynamics and stage-cost in the original formulation of the problem.
\end{abstract}

\section{Introduction}
Consider the following infinite-horizon linear-quadratic (LQ) optimal control problem:
\begin{subequations}\label{eq:opt_problem}
\begin{align}
\min_{u,x} &\sum_{k\in\mathbb{N}_0} x_{k}^{\trs}Q_k x_k + u_k^{\trs}R_k u_k \label{eq:nominal_cost} \\
\intertext{subject to $x_0=\xi$ and}
x_{k+1}&=A_k x_k + B_k u_k, \quad k\in\mathbb{N}_0, \label{eq:system_dynamics}
\end{align}
\end{subequations}
where $Q_k$ is positive semi-definite, $R_k$ is positive definite, $A_k \in \mathbb{R}^{n \times n}$, $B_k \in \mathbb{R}^{n \times m}$, and the initial state $\xi \in \mathbb{R}^n$ is given. The task is to determine the cost minimizing input $u=(u_0,u_1,\ldots)$ and corresponding state sequence $x=(x_0,x_1,\dots)$ over the infinite horizon.

Under assumptions of uniform stabilizability and uniform detectability, it is well-known (e.g., see~\cite{anderson2007optimal,bertsekas2012dynamic}) that the optimal policy for~\eqref{eq:opt_problem} is given by the stabilizing linear time-varying state-feedback controller
\begin{align} \label{eq:optpolicy}
u_k = -(R_k + B_k^{\trs} P_{k+1} B_k)^{-1} B_k^{\trs}P_{k+1}A_k x_k,~ k\in\mathbb{N}_0,
\end{align}
where $P_k$ is the unique positive semi-definite solution of 
\begin{align}\label{eq:recursion}
P_k = \mathcal{R}_k(P_{k+1})
\end{align}
and the Riccati operator is given by
\begin{align}\label{eq:ricc_1}
\mathcal{R}_k(P) := Q_k + A_k^{\trs} \left(P - PB_k(R_k\!+\!B_k^{\trs}PB_k)^{-1}B_k^{\trs}P\right) A_k .
\end{align}
For the infinite-horizon problem, the recursion~\eqref{eq:recursion} does not have a boundary condition. Its unique symmetric positive semi-definite solution is stabilizing and attractive for all symmetric positive semi-definite solutions of~\eqref{eq:recursion} over a finite horizon with suitable boundary conditions~\cite{de1992time}.

The cost associated with the optimal policy (\ref{eq:optpolicy}) is given by 
$\xi^{\trs}P_0\xi$~\cite{anderson2007optimal}. 
By the principle of optimality~\cite{bertsekas2012dynamic}, the least infinite-horizon cost is achieved by the receding finite-horizon control scheme given by the feedback policy
\begin{align*}
    u_k = u_k^*(x_k)
\end{align*}
for $k\in\mathbb{N}_0$,
where
\begin{align}
    &(u_k^*(x_k),\dots,u_{k+T-1}^*(x_k)) \nonumber \\
    &\!= \argmin_{u} \!\! \sum_{l=k}^{k+T-1} \!\!\! x_{l}^{\trs} Q_{l} x_{l} + u_{l}^{\trs} R_{l} u_{l} + 
    x_{k+T}^{\trs} P_{k+T} x_{k+T}
    \label{eq:idealMPC}
    \end{align}
with    
 $x_{l+1} = A_l x_l + B_l u_l$ for $l\in[k:k+T-1]$. 
However, without exact knowledge of the problem data $(A_k,B_k,Q_k,R_k)$ beyond the prediction horizon, one can only approximate the optimal cost-to-go $x_{k+T}^{\trs} P_{k+T} x_{k+T}$ with an alternative terminal penalty. 

In the stationary setting, where $A_k=A$, $B_k=B$, $Q_k=Q$, $R_k=R$, and $\mathcal{R}_k=\mathcal{R}$ in all~$k\in\mathbb{N}_0$,
the optimal cost-to-go matrix $P_{k+T}$ for (\ref{eq:idealMPC}) is the constant stabilizing solution to the corresponding algebraic Riccati equation $P=\mathcal{R}_k(P)=\mathcal{R}(P)$. When it is approximated by any constant positive semi-definite $\tilde{P}$ such that $\mathcal{R}(\tilde{P}) \preceq \tilde{P}$, the resulting control policy remains stabilizing~\cite{bitmead1991riccati}. 
Riccati contraction based analysis of receding horizon schemes appears in~\cite{zhang2021regret,li2022performance}.
In~\cite{li2022performance}, it is shown that the closed-loop performance degradation is bounded in terms of the induced $2$-norm of the approximation error $P-\tilde{P}$.

In the stationary non-linear continuous-time setting of~\cite{jadbabaie2005stability}, it is established that when the approximate terminal penalty is set to zero, there exists a finite prediction horizon that guarantees stability. A similar result appears in~\cite{grune2008infinite}, where a performance bound is also quantified for the zero terminal penalty approximation in a stationary non-linear discrete-time setting. Related time-varying results appear in~\cite{keerthi1988optimal} and~\cite{lin2021perturbation}. In~\cite{keerthi1988optimal}, the performance degradation is analyzed for receding-horizon approximations with a constraint that corresponds to infinite terminal penalty. 
In~\cite{lin2021perturbation}, it is shown that for linear dynamics, a sufficiently long prediction horizon, and a time-invariant terminal penalty, the dynamic regret decays exponentially with respect to the horizon. 


Subsequent consideration of the contraction properties of (\ref{eq:ricc_1}) is motivated by the possibility of using the result of iterating (\ref{eq:recursion}) over a finite horizon for a suitable boundary condition, to approximate $P_{k+T}$ in (\ref{eq:idealMPC}). Riccati contraction informed design of horizon length and terminal penalty, given uncertain problem data, is the topic of ongoing investigation. Initial results are reported in~\cite{sun2023receding}. Here, the main contribution relates to characterizing a bound on the strict contraction rate of a sufficient number of iterations of the Riccati recursion (\ref{eq:recursion}), building upon a foundation result from~\cite{bougerol1993kalman}. The sufficient number of iterations is related to an assumed uniform controllability and observability property of the time-varying dynamics and stage costs. The development involves a lifted reformulation of the problem \eqref{eq:opt_problem}, in which the system model evolves by this fixed number of steps per stage. The fixed number of steps and the corresponding bound on the strict contraction rate are given explicitly in terms of the original problem data.

The paper is organized as follows.
Contraction properties of the Riccati operator with respect to the Riemannian metric are presented in Section~\ref{sec:contraction}. The lifting approach for characterizing the strict contraction rate is developed in Section~\ref{sec:lifting}. A numerical example is presented in Section~\ref{sec:example}. Some concluding remarks are provided in Section~\ref{sec:conclusion}.

\paragraph*{Notation}
$\mathbb{N}$ denotes the set of natural numbers, and $\mathbb{N}_0 = \mathbb{N} \cup \{0\}$. The $n \times n$ identity matrix is denoted by~$I_n$. The $a \times b$ matrix of zeros is denoted by~$0_{a,b}$.
Given the indexed collection of matrices $(M_a, M_{a+1}, \dots, M_{b})$, where $a,b \in \mathbb{N}_0$, and $a < b$, the corresponding block-diagonal matrix is denoted by $\mathop{\oplus}_{j=a}^{b} M_j$. The transpose of the matrix $M$ is denoted by $M^{\trs}$. The induced $2$-norm of $M$ is denoted by~$\|M\|_2$; this corresponds to the maximum singular value. For~$n \in \mathbb{N}$,
the set of $n\times n$ real symmetric matrices is denoted by $\mathbb{S}$, the positive semi-definite matrices by $\mathbb{S}_{+}^{n}\subset\mathbb{S}$, and positive definite matrices by $\mathbb{S}_{++}^{n}\subset\mathbb{S}^{n}_{+}$. The minimum eigenvalue of $M\in\mathbb{S}$ is denoted by~$\lambda_{\min}(M)\in\mathbb{R}$.


\section{Riccati operator contraction properties}\label{sec:contraction}
In this section, a result in~\cite{bougerol1993kalman} is used to establish that $\mathcal{R}_k$ in \eqref{eq:ricc_1} is a contraction with respect to the Riemannian metric on the set of positive definite matrices.

\begin{assumption}\label{asm:a_invertible}
$A_k$ in \eqref{eq:system_dynamics} is non-singular for all $k \in \mathbb{N}_0$.
\end{assumption}

This standing assumption and the following lemma enable access to a foundation result from~\cite{bougerol1993kalman} in subsequent developments.
A proof is given in Appendix~\ref{sec:proof_rewrite_r}.

\begin{lemma}\label{lemma:rewrite_r}
The operator~$\mathcal{R}_k$ in~\eqref{eq:ricc_1} can be written as the linear fractional transformation
\begin{align}\label{eq:ricc_2}
\mathcal{R}_k(P) = (E_k P + F_k)(G_k P + H_k)^{-1} ,
\end{align}
where
\begin{subequations}\label{eq:efgh}
\begin{align}
E_k &= A_k^{\trs} + Q_kA_k^{-1}B_kR_k^{-1}B_k^{\trs} , \label{eq:e_k} \\
F_k &= Q_k A_k^{-1}, \\
G_k &= A_k^{-1} B_k R_k^{-1} B_k^{\trs}, \\
H_k &= A_k^{-1} .
\end{align}
\end{subequations}
\end{lemma}

\begin{definition}
The Riemannian distance between $U,V\in\mathbb{S}^n_{++}$ is given by
\begin{align*}
\delta(U,V) = \left( \sum_{i=1}^{n} \log^{2} \lambda_i \right)^{\frac{1}{2}} ,
\end{align*}
where 
$\lambda_1, \dots, \lambda_n$ are the eigenvalues of $UV^{-1}$.
\end{definition}

Note, $\delta(\cdot,\cdot):\mathbb{S}_{++}^{n}\times\mathbb{S}_{++}^{n}\rightarrow \mathbb{R}$ is a metric~\cite{bougerol1993kalman}.
The following result is taken from~\cite[Theorem~1.7]{bougerol1993kalman}.
\begin{proposition}
\label{prop:contraction}
Consider the operator~$\mathcal{R}_k$ in~\eqref{eq:ricc_2}. 
If the corresponding matrices in \eqref{eq:efgh} are such that $E_k$ is non-singular and $F_kE_k^{\trs}, E_k^{\trs}G_k \in \mathbb{S}_{+}^{n}$, then for any $X, Y \in \mathbb{S}_{++}^{n}$,
\begin{align*}
\delta(\mathcal{R}_k(X), \mathcal{R}_k(Y)) \leq \delta(X,Y) .
\end{align*}
Further, if 
$F_kE_k^{\trs}, E_k^{\trs}G_k \in \mathbb{S}_{++}^{n}$, then for any $X, Y \in \mathbb{S}_{++}^{n}$,
\begin{align}\label{eq:strict_contraction}
\delta(\mathcal{R}_k(X), \mathcal{R}_k(Y)) \leq \rho_k\cdot \delta(X,Y)
\end{align}
with 
$\rho_k = \zeta_k/(\zeta_k + \epsilon_k)<1$,
where
\begin{align}\label{eq:zeta_eps}
\zeta_k = \|(F_kE_k^{\trs})^{-1}\|_2 \quad \text{ and } \quad \epsilon_k = \lambda_{\min}((E_k^{\trs})^{-1} G_k^{\trs}) .
\end{align}
\end{proposition}

Under Assumption~\ref{asm:a_invertible}, since $R_k 
+B_k^{\trs}(A_k^{\trs})^{-1}Q_kA_k^{-1}B_k\in\mathbb{S}^{n}_{++}$, application of the Woodbury matrix identity yields 
\begin{align*}
E_k^{-1}
&= (A_k^{\trs})^{-1} \!-(A_k^{\trs})^{-1}Q_kA_k^{-1}B_k \\ & \qquad\qquad\quad \times(R_k+B_k^{\trs}(A_k^{\trs})^{-1}Q_kA_k^{-1}B_k)^{-1}B_k^{\trs}(A_k^{\trs})^{-1}.
\end{align*}
That is, $E_k$ is non-singular. On the other hand,
\begin{align}\label{eq:fet}
F_kE_k^{\trs} = Q_k + Q_kA_k^{-1}B_kR_k^{-1}B_k^{\trs}(A_k^{\trs})^{-1}Q_k
\end{align}
and
\begin{align}\label{eq:etg}
E_k^{\trs}G_k = B_k(R_k^{-1}+R_k^{-1}B_k^{\trs}(A_k^{\trs})^{-1}Q_kA_k^{-1}B_kR_k^{-1}) 
B_k^{\trs}
\end{align}
are positive semi-definite but not necessarily positive definite. So in view of  Proposition~\ref{prop:contraction} and Lemma~\ref{lemma:rewrite_r}, the operator~$\mathcal{R}_k$ in \eqref{eq:ricc_1} is a contraction, but not necessarily a strict contraction. A sufficient condition for strict contraction follows.

\begin{proposition}\label{prop:sufficient_contraction}
Consider $\mathcal{R}_k$ in~\eqref{eq:ricc_1}. 
If $Q_k\in\mathbb{S}_{++}^n$, and $B_k$ has full row rank, then 
 for any $X, Y \in \mathbb{S}_{++}^{n}$, \eqref{eq:strict_contraction} holds with~${\rho_k = \zeta_k / (\zeta_k + \epsilon_k)}<1$, where
\begin{subequations}\label{eq:new_zeta_eps}
\begin{align}
\zeta_k \!&=\! \|(Q_k + Q_kA_k^{-1}B_kR_k^{-1}B_k^{\trs}(A_k^{\trs})^{-1}Q_k)^{-1}\|_2 , \\
\epsilon_k \!&= \! \lambda_{\min} (A_k^{-1}B_k (R_k\!+\!B_k^{\trs}(A_k^{\trs})\!^{-1}Q_kA_k^{-1}B_k)\!^{-1}\! B_k^{\trs}(A_k^{\trs})\!^{-1} ) . \label{eq:new_eps}
\end{align}
\end{subequations}
\end{proposition}

\begin{proof}
From~\eqref{eq:fet} and~\eqref{eq:etg}, if $Q_k$ is positive definite and $B_k$ has full row rank, then $F_kE_k^{\trs}, E_k^{\trs}G_k \in \mathbb{S}_{++}^{n}$, and the strict contraction properties follow from Proposition~\ref{prop:contraction}. Consider $E_k,F_k,G_k$ in~\eqref{eq:efgh}. Then, \eqref{eq:zeta_eps} leads to~\eqref{eq:new_zeta_eps}. In particular, by application of the Woodbury matrix identity,
\begin{align*}
&(E_k^{\trs})^{-1} G_k^{\trs}\\
&=\! (A_k + B_kR_k^{-1}B_k^{\trs}(A_k^{\trs})^{-1}Q_k)^{-1} B_kR_k^{-1}B_k^{\trs}(A_k^{\trs})^{-1} \\
&=\! (I_n \!+\! A_k^{-1}B_kR_k^{-1}B_k^{\trs}(A_k^{\trs})^{-1}Q_k)^{-1} \! A_k^{-1}B_kR_k^{-1}B_k^{\trs}(A_k^{\trs})^{-1} \\
&=\! (I_n \!-\! A_k^{-1}B_k(R_k\!+\!B_k^{\trs}(A_k^{\trs})\!^{-1}\!Q_kA_k^{-1}B_k)^{-1}\!B_k^{\trs}(A_k^{\trs})\!^{-1}Q_k) \\
&\qquad\qquad\qquad\qquad\qquad\qquad\qquad \times A_k^{-1}B_kR_k^{-1}B_k^{\trs}(A_k^{\trs})^{-1} \\
&=\! A_k^{-1}B_k \Big(I_n \!-\! (R_k\!+\!B_k^{\trs}(A_k^{\trs})^{-1}Q_kA_k^{-1}B_k)^{-1} \\
&\qquad\qquad\qquad\qquad 
\times B_k^{\trs}(A_k^{\trs})^{-1}Q_kA_k^{-1}B_k\Big) R_k^{-1}B_k^{\trs}(A_k^{\trs})^{-1} \\
&=\! A_k^{-1}\!B_k (R_k\!+\!B_k^{\trs}(A_k^{\trs})^{-1}Q_kA_k^{-1}B_k)^{-1}\!R_k R_k^{-1}B_k^{\trs}(A_k^{\trs})^{-1} ,
\end{align*}
which with \eqref{eq:zeta_eps} yields~\eqref{eq:new_eps}.
\end{proof}

\section{Lifting to a strict contraction}\label{sec:lifting}
A lifted reformulation of problem \eqref{eq:opt_problem} is developed below for which the corresponding Riccati operator is strictly contractive. In the lifted representation, each stage of the system model corresponds to multiple steps of~\eqref{eq:system_dynamics}, with a view to satisfying the conditions of Proposition~\ref{prop:sufficient_contraction}. This is achieved under a combined uniform controllability and observability assumption on the original formulation \eqref{eq:opt_problem}. 

Given $d \in \mathbb{N}$, with reference to \eqref{eq:system_dynamics},
define 
the $d$-step lifted model state
\begin{align} \label{eq:lifted_x}
\tilde{x}_t := x_{dt}, 
\end{align}
and input
\begin{align} \label{eq:lifted_u}
\hat{u}_t :=
\begin{bmatrix}
u_{dt}^\trs & u_{dt+1}^\trs & \cdots & u_{d(t+1)-1}^\trs
\end{bmatrix}^\trs
\end{align}
for each $t\in\mathbb{N}_0$. Then,
\begin{align}\label{eq:dynamics_shorthand}
\hat{A}_t
\begin{bmatrix}
x_{dt} \\ \vdots \\ x_{d(t+1)-1} \\ \tilde{x}_{t+1}
\end{bmatrix}
=
\hat{B}_t
\hat{u}_t
+
\begin{bmatrix}
\tilde{x}_t \\ 0_{nd,1}
\end{bmatrix} ,
\end{align}
where
\begin{subequations}  \label{eq:ABhat}
\begin{align}
\hat{A}_t &:= I_{n(d+1)} -
\begin{bmatrix}
0_{n,nd} & 0_{n,n} \\
\mathop{\oplus}_{j=0}^{d-1}A_{dt+j} & 0_{nd,n}
\end{bmatrix} , \label{eq:Ahat}\\
\hat{B}_t &:= \begin{bmatrix}
0_{n,md} \\ \mathop{\oplus}_{j=0}^{d-1}B_{dt+j}
\end{bmatrix} . \label{eq:Bhat}
\end{align}
\end{subequations}
On noting that $\hat{A}_t$ is non-singular for all $t\in\mathbb{N}_0$, the following lemma is a direct consequence of~\eqref{eq:dynamics_shorthand}.

\begin{lemma}
Given input $u$ for the system dynamics \eqref{eq:system_dynamics}, the lifted model state in \eqref{eq:lifted_x} evolves according to 
%
\begin{align}\label{eq:new_dynamics}
\tilde{x}_{t+1}
&= \Phi_t\, \tilde{x}_t + \Gamma_t \hat{u}_t , \quad t\in\mathbb{N}_0,
\end{align}
where the lifted input $\hat{u}$ is as given in \eqref{eq:lifted_u}, and
\begin{subequations}
\begin{align}
\Phi_t &:= \begin{bmatrix}
0_{n,nd} & I_n
\end{bmatrix}
\hat{A}_t^{-1}
\begin{bmatrix}
I_n \\ 0_{nd,n}
\end{bmatrix} , \label{eq:phi_compact} \\
\Gamma_t &:= \begin{bmatrix}
0_{n,nd} & I_n
\end{bmatrix}
\hat{A}_t^{-1} \hat{B}_t \label{eq:gamma_compact},
\end{align}
\end{subequations}
with $\hat{A}_t$ and $\hat{B}_t$ as per \eqref{eq:ABhat}.
\end{lemma}
\begin{remark}
The matrix~$\Gamma_t$ in \eqref{eq:gamma_compact} is the $d$-step controllability matrix for system~\eqref{eq:system_dynamics} in the un-lifted domain.
\end{remark}


For $t\in\mathbb{N}_0$, define $C_t := Q_t^{\frac{1}{2}}$, and given $d\in\mathbb{N}_0$, 
\begin{align} \label{eq:ChatRhat}
\hat{C}_t &:=
\begin{bmatrix}
{\displaystyle \mathop{\oplus}_{j=0}^{d-1}}C_{dt+j} & 0_{nd,n}
\end{bmatrix}~\text{ and }~
\hat{R}_t := \mathop{\oplus}_{j=0}^{d-1}R_{dt+j}.
\end{align}
\begin{lemma}\label{lemma:new_cost}
Given input $u$, the cost in problem \eqref{eq:opt_problem} equals
\begin{align}\label{eq:new_cost}
\sum_{t \in \mathbb{N}_0} \begin{bmatrix} \tilde{x}_t \\ \hat{u}_t \end{bmatrix}^{\trs}
\begin{bmatrix}
\Xi_t^{\trs} \Xi_t & \Xi_t^{\trs}\Delta_t \\
\Delta_t^{\trs}\Xi_t & \hat{R}_t+\Delta_t^{\trs}\Delta_t
\end{bmatrix}
\begin{bmatrix} \tilde{x}_t \\ \hat{u}_t \end{bmatrix},
\end{align}
with $\tilde{x}_t$ as per \eqref{eq:new_dynamics} for the lifted input $\hat{u}$ given in \eqref{eq:lifted_u}, and
\begin{subequations} \label{eq:XiDelta}
\begin{align}
\Xi_t &:= \hat{C}_t \hat{A}_t^{-1} \begin{bmatrix}
I_n \\ 0_{nd,n}
\end{bmatrix} , \label{eq:xi_compact} \\
\Delta_t &:= \hat{C}_t \hat{A}_t^{-1} \hat{B}_t. \label{eq:delta_compact}
\end{align}
\end{subequations}
\end{lemma}

The proof of Lemma~\ref{lemma:new_cost} is deferred to Appendix~\ref{sec:proof_new_cost}.


\begin{remark}
The matrix~$\Xi_t$ in \eqref{eq:xi_compact} is the $d$-step observability matrix for system~\eqref{eq:system_dynamics} in the un-lifted domain.
\end{remark}



Cross-terms appear in the expression~\eqref{eq:new_cost} of the cost in the lifted domain. This is incompatible with the formulation of Proposition~\ref{prop:sufficient_contraction}.
An LDU decomposition and corresponding lifted domain change of variable 
\begin{align}\label{eq:tilde_v}
\tilde{u}_t := (\hat{R}_t + \Delta_t^{\trs}\Delta_t)^{-1}\Delta_t^{\trs}\Xi_t\tilde{x}_t + \hat{u}_t, \quad t\in\mathbb{N}_0,
\end{align}
leads to the following reformulation of problem \eqref{eq:opt_problem} in the required form.

\begin{lemma}\label{lemma:equivalence}
Problem \eqref{eq:opt_problem} is equivalent to the lifted problem
\begin{subequations} \label{eq:lifted_prob}
\begin{align}
\min_{\tilde{x},\tilde{u}} &\sum_{t\in\mathbb{N}_0} \tilde{x}_{t}^{\trs} \tilde{Q}_t \tilde{x}_t + \tilde{u}_t^{\trs} \tilde{R}_t \tilde{u}_t \label{eq:final_cost} \\
\intertext{subject to $\tilde{x}_0 = \xi$ and}
\tilde{x}_{t+1} &= \tilde{A}_t \tilde{x}_t + \tilde{B}_t \tilde{u}_t , \quad t\in\mathbb{N}_0,\label{eq:final_dynamics} 
\end{align}
\end{subequations}
where
\begin{align}
&\tilde{Q}_t := \Xi_t^{\trs}\Xi_t - \Xi_t^{\trs}\Delta_t\tilde{R}_t^{-1}\Delta_t^{\trs}\Xi_t , \label{eq:tilde_q} \\
&\tilde{R}_t := \hat{R}_t + \Delta_t^{\trs}\Delta_t , \label{eq:tilde_r} \\
&\tilde{A}_t := \Phi_t - \Gamma_t \tilde{R}_t^{-1} \Delta_t^{\trs} \Xi_t , \label{eq:tilde_a} \\
&\tilde{B}_t := \Gamma_t . \label{eq:tilde_b}
\end{align}
\end{lemma}

\begin{proof}
The equivalence follows by noting that 
\begin{align}\label{eq:ldu_decomp}
\begin{bmatrix}
\Xi_t^{\trs} \Xi_t & \Xi_t^{\trs}\Delta_t \\
\Delta_t^{\trs}\Xi_t & \tilde{R}_t
\end{bmatrix}
=
L^{\trs}
\begin{bmatrix}
\tilde{Q}_t & 0 \\
0 & \tilde{R}_t
\end{bmatrix}
L ,
\end{align}
where
\begin{align*}
&L := \begin{bmatrix} I_n & 0 \\ \tilde{R}_t^{-1}\Delta_t^{\trs}\Xi_t & I_n \end{bmatrix} .
\end{align*}
With the correspondingly transformed input defined in \eqref{eq:tilde_v}, the cost \eqref{eq:new_cost} becomes the cost in \eqref{eq:lifted_prob}, and the lifted state evolves according to
\begin{align*}
\tilde{x}_{t+1} = \Phi_t \tilde{x}_t + \Gamma_t \left[\tilde{u}_t - (\hat{R}_t + \Delta_t^{\trs}\Delta_t)^{-1}\Delta_t^{\trs}\Xi_t\tilde{x}_t \right] ,
\end{align*}
which is~\eqref{eq:final_dynamics}. As such, $\tilde{x}$ is defined given either $\hat{u}$ or $\tilde{u}$, and either can be constructed from the other using \eqref{eq:tilde_v}.
\end{proof}


\begin{assumption}
\label{asm:uni_ctr_obs}
For all $t \in \mathbb{N}_0$, the $d$-step controllability matrix $\Gamma_t$ in \eqref{eq:gamma_compact} has full row rank, and the $d$-step observability matrix~$\Xi_t$ in \eqref{eq:xi_compact} has full column rank.
\end{assumption}

\begin{lemma}\label{lemma:q_tilde_pd}
With $d \in \mathbb{N}$ such that Assumption~\ref{asm:uni_ctr_obs} holds, the matrix~$\tilde{Q}_t$ in~\eqref{eq:tilde_q} is positive definite for all $t \in \mathbb{N}_0$.
\end{lemma}

\begin{proof}
First observe that application of the Woodbury matrix identity gives
\begin{align}
\tilde{Q}_t 
&= \Xi_t^{\trs} \left(I_{nd} - \Delta_t (\hat{R}_t + \Delta_t^{\trs}\Delta_t)^{-1}\Delta_t^{\trs} \right) \Xi_t \nonumber \\
&= \Xi_t^{\trs} (I_{nd} + \Delta_t \hat{R}_t^{-1} \Delta_t^{\trs})^{-1} \Xi_t. \label{eq:proof_pd_step}
\end{align}
Then note that $(I_{nd} + \Delta_t \hat{R}_t^{-1} \Delta_t^{\trs})^{-1}\in\mathbb{S}_{++}^{n}$. Under Assumption~\ref{asm:uni_ctr_obs}, $\Xi_t$ has full column rank, and thus,  $\tilde{Q}_t\in\mathbb{S}_{++}^n$  in view of \eqref{eq:proof_pd_step}.
\end{proof}

\begin{lemma}\label{lemma:tilde_a_nonsingular}
With $d \in \mathbb{N}$ such that Assumption~\ref{asm:uni_ctr_obs} holds, the state matrix~$\tilde{A}_t$ in \eqref{eq:final_dynamics} is non-singular for all $t \in \mathbb{N}_0$.
\end{lemma}

The proof of Lemma~\ref{lemma:tilde_a_nonsingular} is deferred to Appendix~\ref{sec:tilde_a_nonsingular}.

\begin{theorem}
\label{theorem:contraction_lifted}
With $d \in \mathbb{N}$ such that Assumption~\ref{asm:uni_ctr_obs} holds, for $P \in \mathbb{S}_{++}^{n}$ and $t\in\mathbb{N}_0$, define the Riccati operator
\begin{align}\label{eq:operator_lifted}
\tilde{\mathcal{R}}_t(P) := \tilde{Q}_t + \tilde{A}_t^{\trs} (P - P\tilde{B}_t(\tilde{R}_t+\tilde{B}_t^{\trs}P\tilde{B}_t)^{-1}\tilde{B}_t^{\trs}P) \tilde{A}_t,
\end{align}
 with $\tilde{Q}_t, \tilde{R}_t, \tilde{A}_t, \tilde{B}_t$ as per~\eqref{eq:tilde_q}, \eqref{eq:tilde_r}, \eqref{eq:tilde_a},~\eqref{eq:tilde_b}, respectively. Then, 
for any $X,Y \in \mathbb{S}_{++}^{n}$,
\begin{align}
\delta(\tilde{\mathcal{R}}_t(X), \tilde{\mathcal{R}}_t(Y)) \leq \tilde{\rho}_t \cdot \delta(X,Y),
\end{align}
with $\tilde{\rho}_t = \tilde{\zeta}_t/(\tilde{\zeta}_t + \tilde{\epsilon}_t)< 1$,
where
\begin{subequations}
\begin{align*}
\tilde{\zeta}_t &= \|(\tilde{Q}_t + \tilde{Q}_t\tilde{A}_t^{-1}\tilde{B}_t\tilde{R}_t^{-1}\tilde{B}_t^{\trs}(\tilde{A}_t^{\trs})^{-1}\tilde{Q}_t)^{-1}\|_2 , \\
\tilde{\epsilon}_t &= \lambda_{\min} (\tilde{A}_t^{-1}\tilde{B}_t (\tilde{R}_t\!+\!\tilde{B}_t^{\trs}(\tilde{A}_t^{\trs})\!^{-1}\tilde{Q}_t\tilde{A}_t^{-1}\tilde{B}_t)\!^{-1}\! \tilde{B}_t^{\trs}(\tilde{A}_t^{\trs})\!^{-1} ) .
\end{align*}
\end{subequations}
\end{theorem}

\begin{proof}
Under Assumption~\ref{asm:uni_ctr_obs}, $\tilde{B}_t = \Gamma_t$ has full row rank for all $t \in \mathbb{N}_0$. From Lemma~\ref{lemma:q_tilde_pd} and Lemma~\ref{lemma:tilde_a_nonsingular}, $\tilde{Q}_t$ is positive definite and $\tilde{A}_t$ is invertible for all $t \in \mathbb{N}_0$ in line with Assumption~\ref{asm:a_invertible}. As such, the strict contraction property follows from Proposition~\ref{prop:sufficient_contraction}.
\end{proof}

The lifted Riccati operator $\tilde{\mathcal{R}}_t$ in \eqref{eq:operator_lifted} corresponds to composing the original $\mathcal{R}_k$ in \eqref{eq:ricc_1} according to \eqref{eq:recursion}. 

\begin{proposition}\label{prop:composition}
Given $d \in \mathbb{N}$, for all~$P \in \mathbb{S}_{++}^{n}$ and~$t \in \mathbb{N}_0$,
\begin{align}\label{eq:composition}
\tilde{\mathcal{R}}_t(P) = \mathcal{R}_{dt}\circ\mathcal{R}_{dt+1} \circ \cdots \circ\mathcal{R}_{d(t+1)-1}(P).
\end{align}
\end{proposition}

The proof is deferred to Appendix~\ref{sec:proof_composition}.

\section{Example}\label{sec:example}
A numerical example is presented to illustrate the strict contraction properties of the Riccati operator. Consider the following instance of the time-varying LQ control problem~\eqref{eq:opt_problem}:
For $k \in \mathbb{N}_0$,
\begin{align*}
Q_k &= \begin{bmatrix}10&4\\4&7\end{bmatrix} + \alpha^k \sin(\omega k) \begin{bmatrix}2&1\\1&3\end{bmatrix} , \\
R_k &= 5 + 4 \alpha^k \sin(\omega k), \\
A_k &= \begin{bmatrix}5&3\\2&1\end{bmatrix} + \alpha^k \sin(\omega k) \begin{bmatrix}10&20\\30&10\end{bmatrix} , \\
B_k &= \begin{bmatrix}2\\3\end{bmatrix} + \alpha^k \sin(\omega k) \begin{bmatrix}10\\20\end{bmatrix} ,
\end{align*}
where $\alpha = 0.9$, and $\omega = 1$.

The time-varying dynamics are uniformly $d$-step controllable and observable in the sense of Assumption~\ref{asm:uni_ctr_obs} for $d = 2$. Consider the corresponding Riccati recursions \begin{align*}
    X_k = \mathcal{R}_k(X_{k+1}) \quad \text{and} \quad
    Y_k = \mathcal{R}_k(Y_{k+1})
\end{align*}
for $k = T-1, T-2, \ldots, 0$, with boundary conditions
\begin{align*}
X_T = 10^{-2} \cdot I_2 \quad \text{and} \quad Y_T = 10^{2} \cdot I_2. 
\end{align*}

With~$T=20$, the distance between $X_k$ and $Y_k$ is measured by the Riemannian distance $\delta(X_k,Y_k)$ and the induced $2$-norm $\|X_k - Y_k\|_2$, respectively. The results are plotted in Figure~\ref{fig:distance}.

\begin{figure}[h]
    \centering
    \includegraphics[width=\linewidth]{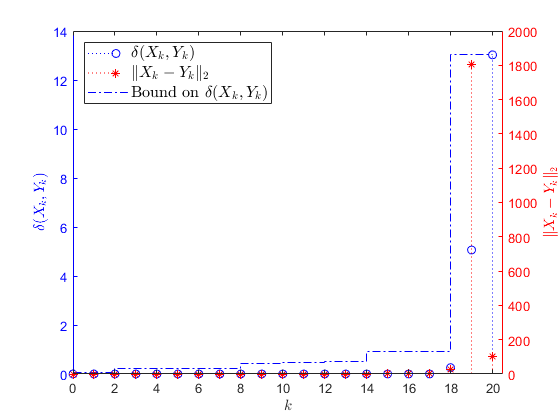}
    \caption{Distance between $X_k$ and $Y_k$.}
    \label{fig:distance}
\end{figure}

Given the uniform controllability and observability index~$d=2$, the system model in the lifted reformulation of the problem evolves by $2$ steps per stage. According to Theorem~\ref{theorem:contraction_lifted}, the Riccati operator in the lifted domain is strictly contractive with respect to the Riemannian distance, with time-varying rate of contraction, as shown in Figure\ref{fig:distance}.

Observe from Figure~\ref{fig:distance} that the Riccati operator is not initially a  contraction with respect to the induced $2$-norm.

\section{Conclusion}\label{sec:conclusion}
Our attention is focused on the non-stationary Riccati operator associated with the time-varying LQ control problem. The lifting approach presented in this paper provides a procedure to measure the strict contraction rate of the non-stationary Riccati operator. Further extensions to the results in this paper may be possible by replacing the controllability and observability assumptions with weaker assumptions such as stabilizability and detectability. Future work is focused on the impact of error in the cost-to-go approximations on the performance of the receding horizon scheme.

\printbibliography

\appendices

\section{Proof of Lemma~\ref{lemma:rewrite_r}}\label{sec:proof_rewrite_r}
First note that 
\begin{align*}
\mathcal{R}_k(P)&= Q_k + A_k^{\trs} \left(P - PB_k(R_k\!+\!B_k^{\trs}PB_k)^{-1}B_k^{\trs}P\right) A_k\\
&= Q_k + A_k^{\trs}P\left(I\!-\!B_k(R_k\!+\!B_k^{\trs}PB_k)^{-1}B_k^{\trs}P\right) A_k\\
&= Q_k + A_k^{\trs}P(I + B_kR_k^{-1}B_k^{\trs}P)^{-1} A_k,
\end{align*}
where the last equality holds by an application of the Woodbury matrix identity. As such, it follows that
\begin{align*}
\mathcal{R}_k(P)
&= Q_k + A_k^{\trs}P(A_k^{-1} + A_k^{-1}B_k R_k^{-1} B_k^{\trs} P)^{-1}\\
&= \left(Q_k(A_k^{-1} + A_k^{-1}B_k R_k^{-1} B_k^{\trs} P) + A_k^{\trs}P\right) \\
&\mathrel{\phantom{=}} \quad \times (A_k^{-1} + A_k^{-1}B_k R_k^{-1} B_k^{\trs} P)^{-1} \\
&= \left((A_k^{\trs}+Q_kA_k^{-1}B_kR_k^{-1}B_k^{\trs})P + Q_kA_k^{-1} \right) \\
&\mathrel{\phantom{=}} \quad \times ( A_k^{-1}B_k R_k^{-1} B_k^{\trs} P + A_k^{-1})^{-1}
\end{align*}
in accordance with~\eqref{eq:ricc_2} and~\eqref{eq:efgh}.
\qed

\section{Proof of Lemma~\ref{lemma:new_cost}}\label{sec:proof_new_cost}
With reference to \eqref{eq:opt_problem}, \eqref{eq:lifted_x}, and \eqref{eq:ChatRhat}, let 
\begin{align*}
w_t := \hat{C}_t \begin{bmatrix}
x_{dt} \\ \vdots \\ x_{d(t+1)-1} \\ \tilde{x}_{t+1}
\end{bmatrix} 
\end{align*}
for each $t\in\mathbb{N}_0$. Then, in view of \eqref{eq:lifted_u}, \eqref{eq:dynamics_shorthand}, and \eqref{eq:XiDelta}, 
\begin{align*}
w_t = \hat{C}_t \hat{A}_t^{-1} \left( \hat{B}_t \hat{u}_t +
\begin{bmatrix}
\tilde{x}_t \\ 0_{nd,1}
\end{bmatrix} \right) = \Xi_t \tilde{x}_t + \Delta_t \hat{u}_t ,
\end{align*}
and the cost in  \eqref{eq:opt_problem} for the given $u$ can be written as
\begin{align*}
&\sum_{t \in \mathbb{N}_0} w_t^{\trs} w_t + \hat{u}_t^{\trs} \hat{R}_t \hat{u}_t \\
&=\!\sum_{t \in \mathbb{N}_0} \! \tilde{x}_t^{\trs} \Xi_t^{\trs} \Xi_t \tilde{x}_t \!+\! \hat{u}_t^{\trs} (\hat{R}_t \!+\! \Delta_t^{\trs}\Delta_t) \hat{u}_t \!+\! \tilde{x}_t^{\trs}\Xi_t^{\trs}\Delta_t \hat{u}_t \!+\! \hat{u}_t^{\trs}\Delta_t^{\trs}\Xi_t\tilde{x}_t .
\end{align*}
This is~\eqref{eq:new_cost}.
\qed

\section{Proof of Lemma~\ref{lemma:tilde_a_nonsingular}}\label{sec:tilde_a_nonsingular}
Let
\begin{align*}
M := \begin{bmatrix}
\Phi_t & \Gamma_t \\
\Delta_t^{\trs}\Xi_t & \hat{R}_t+\Delta_t^{\trs}\Delta_t
\end{bmatrix} .
\end{align*}
The $22$-block of~$M$ is positive definite and invertible, and its Schur complement is given by
\begin{align} \label{eq:Atil}
\tilde{A}_t := \Phi_t - \Gamma_t (\hat{R}_t + \Delta_t^{\trs}\Delta_t)^{-1} \Delta_t^{\trs} \Xi_t.
\end{align}
Under Assumption~\ref{asm:a_invertible}, the $11$-block
\begin{align*}
\Phi_t = A_{d-1} A_{d-2} \cdots A_0
\end{align*}
is also invertible, and its Schur complement is
\begin{align}\label{eq:schur}
\hat{R}_t + \Delta_t^{\trs}\Delta_t - \Delta_t^{\trs}\Xi_t\Phi_t^{-1}\Gamma_t.
\end{align}
Note that invertibility of \eqref{eq:schur} is equivalent to invertibility of $M$, and thus, invertibility of $\tilde{A}_t$ in \eqref{eq:Atil}. 
By exploiting the structure of 
\begin{align}
&\Delta_t^{\trs}\Delta_t - \Delta_t^{\trs}\Xi_t\Phi_t^{-1}\Gamma_t \nonumber \\
&= \! \hat{B}_t^{\trs}(\hat{A}_t^{\trs})^{-1}\hat{C}_t^{\trs}\hat{C}_t\hat{A}_t^{-1}\hat{B} \nonumber \\
&\quad - \hat{B}_t^{\trs}(\hat{A}_t^{\trs})^{-1}\hat{C}_t^{\trs}\hat{C}_t\hat{A}_t^{-1}\!
\begin{bmatrix} I_n \\ 0_{nd,n} \end{bmatrix} \Phi_t^{-1} \begin{bmatrix} 0_{nd,n} \\ I_n \end{bmatrix}^{\trs}\!\!  \hat{A}_t^{-1}\hat{B}_t \nonumber \\
& = \! \hat{B}_t^{\trs}(\hat{A}_t^{\trs})^{-1}\hat{C}_t^{\trs}\hat{C}_t ( \hat{A}_t^{-1}\!\! -\! \hat{A}_t^{-1}\!
\begin{bmatrix} I_n \\ 0_{nd,n} \end{bmatrix} \Phi_t^{-1} \begin{bmatrix} 0_{nd,n} \\ I_n \end{bmatrix}^{\trs} \!\! \hat{A}_t^{-1} ) \hat{B}_t , \label{eq:schur_psd}
\end{align}
where~\eqref{eq:phi_compact}, \eqref{eq:gamma_compact}, and \eqref{eq:XiDelta}, have been used to arrive at the expression~\eqref{eq:schur_psd}, it can be shown that \eqref{eq:schur} is non-singular. In particular, the structure of \eqref{eq:schur_psd} is block upper-triangular with zero diagonal blocks, and therefore, \eqref{eq:schur} is block upper-triangular, with positive definite diagonal blocks corresponding to those of $\hat{R}_t$. 

First, define the following auxiliary objects
\begin{align*}
\tilde{Y} &:=
\begin{bmatrix}
A_0^{\trs} & (A_1A_0)^{\trs} & \dots & (A_{d-2} \cdots A_0)^{\trs}
\end{bmatrix}^{\trs} , \\
\tilde{X} &:=
\begin{bmatrix}
A_{d-1} \cdots A_1 & A_{d-1} \cdots A_2 & \dots & A_{d-1}
\end{bmatrix} , \\
\tilde{Z} &:=
\begin{bmatrix}
I_n & 0 & \dots & \dots & 0 \\
A_1 & I_n & \ddots &  & \vdots \\
A_2A_1 & A_2 & \ddots & \ddots & \vdots \\
\vdots & \vdots & \ddots & \ddots & 0 \\
A_{d-2} \cdots A_1 & A_{d-2} \cdots A_2 & \dots & A_{d-2} & I_n
\end{bmatrix} .
\end{align*}
The matrix~$\hat{A}_t^{-1}$ is block lower triangular, and all diagonal blocks 
are~$I_n$. In particular,
\begin{align*}
\hat{A}_t^{-1} = \begin{bmatrix}
Y & Z \\
\Phi_t & X
\end{bmatrix} ,
\end{align*}
where
\begin{align*}
Y &:= \begin{bmatrix}
I_n \\ \tilde{Y}
\end{bmatrix} , \\
X &:= \begin{bmatrix}
\tilde{X} & I_n
\end{bmatrix} , \\
Z &:= \begin{bmatrix}
0_{n,nd} \\ \begin{bmatrix} \tilde{Z} & 0_{n(d-1),n} \end{bmatrix}
\end{bmatrix} .
\end{align*}
Therefore,
\begin{align*}
\hat{A}_t^{-1}
\begin{bmatrix} I_n \\ 0_{nd,n} \end{bmatrix} \Phi_t^{-1} \begin{bmatrix} 0_{nd,n} \\ I_n \end{bmatrix}^{\trs} \hat{A}_t^{-1} 
&= \begin{bmatrix}
Y \\ \Phi_t
\end{bmatrix}
\Phi_t^{-1}
\begin{bmatrix}
\Phi_t & X
\end{bmatrix} \\
&= \begin{bmatrix}
Y & Y \Phi_t^{-1} X \\
\Phi_t & X
\end{bmatrix} ,
\end{align*}
and
\begin{align}
&\hat{A}_t^{-1} - \hat{A}_t^{-1}
\begin{bmatrix} I_n \\ 0_{nd,n} \end{bmatrix} \Phi_t^{-1} \begin{bmatrix} 0_{nd,n} \\ I_n \end{bmatrix}^{\trs} \hat{A}_t^{-1} \nonumber\\
&\qquad\qquad =\begin{bmatrix}
0_{nd,n} & Z-Y\Phi_t^{-1}X \\
0_{n,n} & 0_{n,nd}
\end{bmatrix} . \label{eq:auxiliary_up_tri}
\end{align}
Note that
\begin{align*}
Y\Phi_t^{-1}X 
&= \begin{bmatrix}
I_n \\ \tilde{Y}
\end{bmatrix}
\Phi_t^{-1}
\begin{bmatrix}
\tilde{X} & I_n
\end{bmatrix} \\
&= \begin{bmatrix}
I_n \\ A_0 \\ A_1A_0 \\ \vdots \\ A_{d-2} \cdots A_0
\end{bmatrix}
A_0^{-1}A_1^{-1} \cdots A_{d-1}^{-1} \\
&\times
\begin{bmatrix}
A_{d-1} \cdots A_1 & A_{d-1} \cdots A_2 & \dots & A_{d-1} & I_n
\end{bmatrix} \\
&= Z + U ,
\end{align*}
where
\begin{align*}
U =
\begin{bmatrix}
A_0^{-1} & A_0^{-1}A_1^{-1} & \dots & A_0^{-1} \cdots A_{d-1}^{-1} \\
0 & A_1^{-1} & \dots & A_1^{-1} \cdots A_{d-1}^{-1} \\
\vdots & \ddots & \ddots & \vdots \\
0 & \dots & 0 & A_{d-1}^{-1}
\end{bmatrix}
\end{align*}
is block upper triangular.
With~\eqref{eq:auxiliary_up_tri}, it follows that
\begin{align*}
\hat{A}_t^{-1} - \hat{A}_t^{-1}
\begin{bmatrix} I_n \\ 0_{nd,n} \end{bmatrix} \Phi_t^{-1} \begin{bmatrix} 0_{nd,n} \\ I_n \end{bmatrix}^{\trs} \hat{A}_t^{-1}
= \begin{bmatrix}
0_{nd,n} & -U \\
0_{n,n} & 0_{n,nd}
\end{bmatrix},
\end{align*}
and therefore,
\begin{align}
&\hat{C}_t \left( \hat{A}_t^{-1} - \hat{A}_t^{-1}
\begin{bmatrix} I_n \\ 0_{nd,n} \end{bmatrix} \Phi_t^{-1} \begin{bmatrix} 0_{nd,n} \\ I_n \end{bmatrix}^{\trs} \hat{A}_t^{-1} \right) \hat{B}_t \nonumber \\
&= \begin{bmatrix}
\mathop{\oplus}_{j=0}^{d-1}C_{dt+j} & 0_{nd,n}
\end{bmatrix}
\begin{bmatrix}
0_{nd,n} & -U \\
0_{n,n} & 0_{n,nd}
\end{bmatrix}
\begin{bmatrix}
0_{n,md} \\ \mathop{\oplus}_{j=0}^{d-1}B_{dt+j}
\end{bmatrix} , \label{eq:auxiliary_2}
\end{align}
is block upper triangular.
Further,
\begin{align*}
&\hat{B}_t^{\trs} (\hat{A}_t^{\trs})^{-1} \hat{C}^{\trs} \\
&= \begin{bmatrix}
0_{n,md} \\ \mathop{\oplus}_{j=0}^{d-1}B_{dt+j}
\end{bmatrix}
\begin{bmatrix}
Y^{\trs} & \Phi_t^{\trs} \\
Z^{\trs} & X^{\trs}
\end{bmatrix}
\begin{bmatrix}
\mathop{\oplus}_{j=0}^{d-1}C_{dt+j} & 0_{nd,n}
\end{bmatrix} \\
&= \left(\mathop{\oplus}_{j=0}^{d-1}B_{dt+j}\right) Z^{\trs} \left(\mathop{\oplus}_{j=0}^{d-1}C_{dt+j}\right) \\
&= \left(\mathop{\oplus}_{j=0}^{d-1}B_{dt+j}\right)
\begin{bmatrix}
0_{nd,n} & \begin{bmatrix} \tilde{Z}^{\trs} \\ 0_{n,n(d-1)} \end{bmatrix}
\end{bmatrix}
\left(\mathop{\oplus}_{j=0}^{d-1}C_{dt+j}\right) ,
\end{align*}
which is again block upper triangular, now with all zero diagonal blocks.
In conjunction with~\eqref{eq:auxiliary_2}, it follows that~\eqref{eq:schur_psd} is the product of two block upper-triangular matrices, one with zero matrices on the main diagonal. Therefore, \eqref{eq:schur_psd} is also block upper-triangular, with all zero diagonal blocks, as claimed above.
\qed

\section{Proof of Proposition~\ref{prop:composition}}\label{sec:proof_composition}

Given $\xi_t\in\mathbb{R}^n$, consider the finite-horizon LQ control problem
\begin{align}
\min_{u,x} &\sum_{k=dt}^{d(t+1)-1} x_{k}^{\trs}Q_k x_k + u_k^{\trs}R_k u_k + x_{d(t+1)}^{\trs} P x_{d(t+1)} \label{eq:cost_d} \\
\intertext{subject to $x_{dt}=\xi_t$ and}
x_{k+1}&=A_k x_k + B_k u_k, \quad k\in[dt : d(t+1)-1] . \nonumber
\end{align}
The optimal cost~\eqref{eq:cost_d} is given by
\begin{align*}
    V(\xi_t) = \xi_t^{\trs} \left( \mathcal{R}_{dt}\circ\mathcal{R}_{dt+1} \circ \cdots \circ\mathcal{R}_{d(t+1)-1}(P) \right) \xi_t ;
\end{align*}
see~\cite{anderson2007optimal,bertsekas2012dynamic}.
Now, in the vein of Lemma~\ref{lemma:equivalence}, it can be shown that the optimal cost \eqref{eq:cost_d} is equal to
\begin{align}
\min_{\tilde{x},\tilde{u}} \ & \tilde{x}_{t}^{\trs} \tilde{Q}_t \tilde{x}_t + \tilde{v}_t^{\trs} \tilde{R}_t \tilde{u}_t + \tilde{x}_{t+1}^{\trs} P \tilde{x}_{t+1} \label{eq:cost_lifted} \\
\intertext{subject to $\tilde{x}_t = \xi_t$ and}
\tilde{x}_{t+1} &= \tilde{A}_t \tilde{x}_t + \tilde{B}_t \tilde{u}_t . 
\end{align}
In turn, the optimal cost~\eqref{eq:cost_lifted} is given by $V(\xi_t)=\xi_t^{\trs} \tilde{\mathcal{R}}_t(P) \xi_t$. Therefore, the result holds as $\xi_t$ is arbitrary above. 
\qed

\end{document}